\documentclass[conference]{IEEEtran}
\IEEEoverridecommandlockouts

\usepackage[switch]{lineno}

\usepackage{cite}
\usepackage{amsmath,amssymb,amsfonts}
\usepackage{algorithm,algpseudocode}
\usepackage{graphicx}
\usepackage{textcomp}
\usepackage{xcolor}
\usepackage{amsthm}
\usepackage{hyperref}
\usepackage{csvsimple}
\usepackage{comment}
\usepackage{float}
\usepackage[normalem]{ulem}
\usepackage{orcidlink}
\usepackage{caption}
\usepackage{multirow}

\usepackage{tikz}
\usetikzlibrary{automata, positioning}

\usepackage[inline]{enumitem}

\newtheorem{example}{Example}
\newtheorem{lemma}{Lemma}

\newtheorem{definition}{Definition}
\newtheorem{assumption}{Assumption}
\newtheorem{proposition}{Proposition}

\DeclareMathOperator{\Pref}{Pref}

\newcommand{\gera}[1]{{\textcolor{purple}{[Gera: #1]}}}

\newcommand{\aref}[1]{\hyperref[#1]{Appendix~\ref*{#1}}}

\newcommand{\Login}{\ensuremath{\text{\scriptsize\textsf{Login}}}}
\newcommand{\Logout}{\ensuremath{\text{\scriptsize\textsf{Logout}}}}
\newcommand{\AddToCart}{\ensuremath{\text{\scriptsize\textsf{AddToCart}}}}
\newcommand{\RemoveFromCart}{\ensuremath{\text{\scriptsize\textsf{RemoveFromCart}}}}
\newcommand{\Checkout}{\ensuremath{\text{\scriptsize\textsf{Checkout}}}}
\newcommand{\StartSession}{\ensuremath{\text{\scriptsize\textsf{StartSession}}}}

\DeclareMathOperator{\UNR}{UNR}
\DeclareMathOperator{\FDR}{FDR}
\DeclareMathOperator{\ADR}{ADR}
\newcommand{\tuple}[1]{({#1})}

\def\BibTeX{{\rm B\kern-.05em{\sc i\kern-.025em b}\kern-.08em
    T\kern-.1667em\lower.7ex\hbox{E}\kern-.125emX}}
\begin{document}

\makeatletter 
\newcommand{\linebreakand}{%
  \end{@IEEEauthorhalign}
  \hfill\mbox{}\par
  \mbox{}\hfill\begin{@IEEEauthorhalign}
}
\makeatother 


\widowpenalty20000
\clubpenalty20000

\title{Active automata learning and testing techniques for obtaining concise bug descriptions}
\title{Obtaining Concise Bug Descriptions Using Automata Learning and Testing Techniques}
\title{Model-Based Derivation of Concise Bug Descriptions Using Automata Learning and Testing Techniques}
\title{Automata Models for Effective\\Bug Pattern Description}





\author{
    \IEEEauthorblockN{
        Tom Yaacov\IEEEauthorrefmark{1}\orcidlink{0000-0002-0565-6506}, 
        Gera Weiss\IEEEauthorrefmark{1}\orcidlink{0000-0002-5832-8768}, 
        Gal Amram\IEEEauthorrefmark{2}\orcidlink{0000-0003-2138-7542}, and 
        Avi Hayoun\IEEEauthorrefmark{1}\orcidlink{0009-0009-6607-2158}}
    \IEEEauthorblockA{
        \begin{tabular}{cc}
            \begin{tabular}{@{}c@{}}
            \\
                \IEEEauthorrefmark{1}
                    \textit{Ben-Gurion University of the Negev}\\
                   \{tomya, hayounav\}@post.bgu.ac.il, geraw@cs.bgu.ac.il
            \end{tabular} & \begin{tabular}{@{}c@{}}
            \\
                \IEEEauthorrefmark{2}
                    \textit{IBM Research}\\
                    gal.amram@ibm.com
            \end{tabular}
        \end{tabular}
    }
}

\maketitle

\begin{abstract}
Debugging complex systems is a crucial yet time-consuming task. This paper presents the use of automata learning and testing techniques to obtain concise and informative bug descriptions. We introduce the concepts of Failure Explanations (FE), Eventual Failure Explanations (EFE), and Early Detection (ED) to provide meaningful summaries of failing behavior patterns. By factoring out irrelevant information and focusing on essential test patterns, our approach aims to enhance bug detection and understanding. We evaluate our methods using various test patterns and real-world benchmarks, demonstrating their effectiveness in producing compact and informative bug descriptions. 
\end{abstract}

\newcommand{\badge}[3]{
	\ifthenelse{\equal{#2}{}}{}{
		\begin{tikzpicture}[overlay, remember picture]
			\node[xshift=-3.2cm,yshift=-1.1cm] at (current page.north east) {\includegraphics[width=1.8cm]{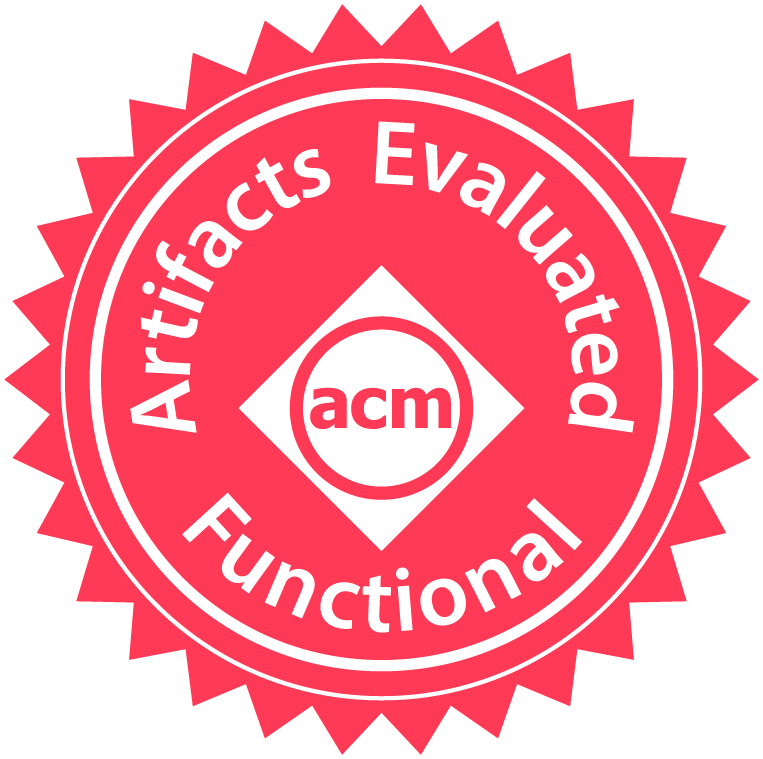}};
		\end{tikzpicture}}
	\ifthenelse{\equal{#2}{}}{}{
		\begin{tikzpicture}[overlay, remember picture]
			\node[xshift=-1.2cm,yshift=-1.1cm] at (current page.north east) {\includegraphics[width=1.8cm]{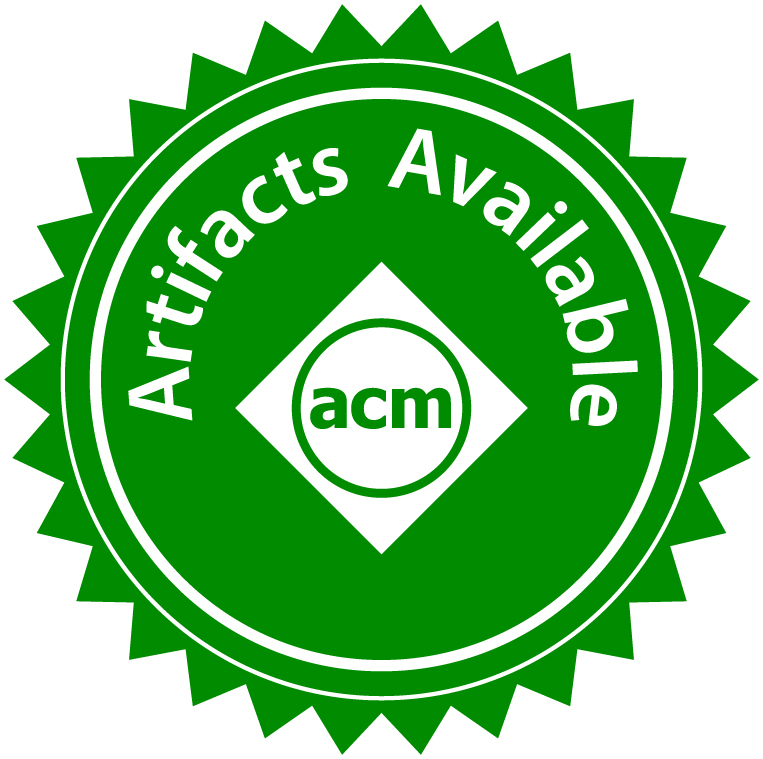}};
		\end{tikzpicture}}
}


\section{Introduction}
Debugging complex systems is a critical yet challenging task. Traditional debugging methods often involve manually identifying patterns of method executions or API calls that seem related to an exhibited unwanted behavior. Moreover, such methods often produce outputs that obscure relevant information about failures. This makes the process not only time-consuming but also prone to human errors. 

Our work is in model-based testing, where a formal model guides test generation and execution. We assume a model that defines the space of possible tests and a set of known failing tests. Each test is a sequence of actions that captures the core testing instructions while abstracting away implementation details. We also assume a testing infrastructure, in the form of an automaton, that translates these abstract sequences into concrete system interactions. For example, a letter representing ``AddToCart'' might expand into multiple API calls that maintain the session state, handle authentication, and verify proper error handling. This separation allows us to focus on the essential testing logic while delegating the complexities of system interaction to the infrastructure layer.

Our approach actively samples the system under test (SUT) to iteratively refine our understanding of failing behaviors. Instead of exhaustively exploring test sequences, we apply automata learning to identify patterns in failures, enabling us to generalize from observed cases to a precise description of the bug's conditions.

We aim to produce automata that meaningfully summarize the set of failed tests. These automata serve multiple purposes in the testing process:
\begin{enumerate}
    \item As a diagnostic tool, they capture the essential patterns that trigger bugs while abstracting away irrelevant details. Rather than examining individual failing tests, testers can analyze the automaton's structure to understand the core sequences of actions that lead to failures. 
    \item As a test generation aid, the automata can be used to systematically produce new test cases that are likely to expose the bug. By traversing different paths through accepting states, testers can generate variations of the failing scenario to better understand the bug's scope.  
    \item As documentation artifacts, they provide a formal yet intuitive representation of bug patterns that can be shared across development teams. Unlike textual bug reports that may be ambiguous, automata precisely specify the problematic behaviors.
    \item As regression testing guides, they can help ensure that fixes actually address the root cause by defining clear acceptance criteria. A proper fix should prevent sequences accepted by the automaton from resulting in failure.
\end{enumerate}

The main challenge lies in generating automata that effectively balance generality and specificity---capturing the essential patterns that lead to bugs while avoiding false positives (i.e., passing tests classified as failures). Applying automata learning algorithms such as \(L^*\) to accept all failing tests and reject all passing ones, as well as invalid tests, may appear simple. However, this approach often leads to the creation of overly complex models that encode system-specific details developers are already familiar with, obscuring the underlying cause of failure rather than clarifying it.

This paper proposes automata constructs that abstract and generalize failure patterns. This allows us to identify key behaviors leading to failures while filtering out irrelevant system details. We introduce three concepts:
\begin{itemize}
    \item \emph{Failure Explanations} (FE): Automata that distinguish sequences characterizing the essence of a failing behavior, helping engineers understand the core conditions under which failures occur.
    \item \emph{Eventual Failure Explanations} (EFE): Automata that accept sequences which, while initially correct, inevitably lead to a failure state, offering insight into latent issues that manifest over time.
    \item \emph{Early Detection} (ED): A heuristic for identifying problematic patterns at their earliest occurrence, before they manifest as observable failures. 
\end{itemize}
Combining these approaches provides complementary perspectives of bug patterns.

The running examples in this paper are based on an experiment testing an online shop~\cite{bar-sinai_provengo_2023}, focusing on session handling, login/logout, and cart operations. Consider the bug: ``Adding an item to the cart, then removing it twice, does not produce the expected error message''. Explaining such bugs is challenging because the test model enforces a structured sequence. For example, tests always begin with opening a session and logging in and end with checkout. This structure obscures bug explanations, as a naive description includes these details, despite offering no real value to developers debugging the issue.


As demonstrated in our experiments, our methods are intended to support an iterative workflow in which engineers refine the test model until they derive explanations that best suit their needs. This adaptive approach facilitates the diagnosis of complex software faults by emphasizing the structural properties of failures rather than incidental details of individual test executions.

The paper is structured as follows: In \autoref{sec:preliminaries}, we establish the formal framework and notations used throughout the paper. \autoref{sec:setting} presents our formal problem setting. \autoref{sec:example} introduces the running example that illustrates our key concepts throughout the paper. We then present our core theoretical contributions: failure explanations in \autoref{sec:fe}, eventual failure explanations in \autoref{sec:efe}, and early detection in \autoref{sec:early-detection}. \autoref{sec:algorithm} describes our algorithmic approach for learning these explanations. \autoref{sec:evaluation} evaluates our techniques on real-world systems and a short discussion on the results. We conclude with a review of related work in \autoref{sec:related-work}.

\section{Preliminaries} 
\label{sec:preliminaries}

An \emph{alphabet} $\Sigma$ is a finite set of letters. A \textit{word} is a finite sequence $w = \sigma_1\dots \sigma_n$ of letters $\sigma_i \in \Sigma$. $\Sigma^*$ is the set of all words over $\Sigma$, while a \emph{language} $L$ is a subset of $\Sigma^*$. The concatenation of two languages $L_1$ and $L_2$, denoted $L_1L_2$, is the set of all words $w_1w_2$, where $w_1 \in L_1$ and $w_2 \in L_2$. We also use concatenation between a word and a language similarly: $wL = \{w\}L$. For $L\subseteq L'$, $L$ is \emph{extension closed with respect to $L'$} if for all $w\in L$ and $wu\in L'$, it holds that $wu\in L$. $L$ is extension closed if it is extension closed with respect to $\Sigma^*$. The set of prefixes of $L$, denoted by $\Pref(L)$, is the language $\{w \in \Sigma^*: \exists u\in \Sigma^* \text{ such that }  wu\in L\}$.

A deterministic finite automaton (DFA) $\mathcal{A}$ is a five-tuple $(Q, \Sigma, q_0, \delta, F)$, where $Q$ is a finite set of states, $q_0 \in Q$ is the initial state, $\Sigma$ is an alphabet, $\delta \colon Q \times \Sigma \to Q$ is the transition function, and $F \subseteq Q$ is the set of accepting states. The language accepted by the automaton $\mathcal{A}$ is denoted by $L(\mathcal{A})$, and is also referred to as $A$.

Throughout the paper, we may refer to a DFA and its language interchangeably. For instance, instead of explicitly stating that \( L(\mathcal{A}) \) is extension closed, we may simply say that the DFA \( \mathcal{A} \) is extension closed.


Automata learning is the process of learning an automaton from a set of observations. Automata learning techniques are categorized in two ways: \emph{active learning}, in which the system under learning is available for querying, and \emph{passive learning}, which uses a static dataset of words that are labeled according to their membership in the language of the automaton being learned.

The $L^*$ algorithm, developed by Angluin~\cite{angluin_learning_1987}, introduces a framework for active learning of DFAs in the presence of a teacher that answers two types of queries:

\begin{itemize}
    \item \emph{Membership queries}, which are used to determine whether a given word $w \in \Sigma^*$ is in $U$.
    \item \emph{Equivalence queries}, which are used to check if the language accepted by a given hypothesis DFA $\mathcal{A}$ is equal to the target DFA.
\end{itemize}

Essentially, $L^*$ learns an unknown regular language $U$ by iteratively constructing a minimal DFA $\mathcal{A}$, using responses of membership and equivalence queries, such that $L(\mathcal{A})$ converges to $U$ throughout the process.


The \emph{Regular Positive and Negative Inference} (RPNI) algorithm~\cite{oncina_identifying_1992} is a polynomial-time passive learning algorithm for the identification of a DFA consistent with a given dataset of positive words $S^+$, and negative words $S^-$. The algorithm constructs a small automaton that accept all $S^+$ members, and rejects all $S^-$ members.
For a detailed description of the algorithm, see~\cite{de_la_higuera_grammatical_2010}.

\section{Formal Setting of the Problem}
\label{sec:setting}

We propose using a test model \( T \) as a regular language that defines the test space. Testers are expected to implement a testing infrastructure capable of executing tests (words) in \( T \) and classifying them as either \emph{passed}, \emph{failed}, or \emph{invalid} (i.e., tests that cannot be executed). In this general framework, each letter in a word represents a test step, which may involve context-dependent validations. The infrastructure handles implementation details, such as mitigating stochastic system behavior by executing tests multiple times. An example of such an infrastructure is described by Bar-Sinai et al.~\cite{bar-sinai_provengo_2023}.

We use $S \subseteq T$ to denote the set of tests the system can execute. The set of failed tests, called bugs, is denoted as $B \subseteq S$. We assume that $B$ is extension-closed with respect to $S$, representing the notion that once the system reaches an error state, any continuation will also be erroneous. This assumption is common in testing, as failures typically persist and continue to affect the system.

\begin{assumption}  
\label{ass:test-model}  
We assume that \( T \), \( S \), and \( B \) satisfy:  
\begin{itemize}  
   \item \( T \subseteq \Sigma^* \) is a regular language over the alphabet \( \Sigma \).  
   \item For any \( w \in T \), we can execute \( w \) and determine whether \( w \in S \), and if so, whether \( w \in B \).  
   \item $B$ is extension closed with respect to $S$, meaning that if \( w \in B \) and \( wu \in S \), then \( wu \in B \).
\end{itemize}  
\end{assumption}

It is important to recognize that while we use methods to learn regular languages, bug descriptions themselves are often not regular. Real-world errors may involve non-regular features such as counting or complex data dependencies. 
Our proposed techniques handle such cases by producing regular languages that approximate the bug and include test sequences that are not strictly feasible in practice.  
Therefore, we manage to employ formal means to cope with arbitrary languages. This allows for a practical balance between expressiveness and tractability in bug detection.



Typically, the test space omits possible interactions with the SUT, reflecting the testers' focus on a specific subsystem or functionality. The test model is often designed to target a particular area of interest, such as a recently reported bug or a critical feature requiring thorough validation. By constraining the test space in this way, testers can efficiently allocate resources and increase the likelihood of detecting relevant issues while maintaining a manageable scope for analysis. At the same time, the test space may over-approximate the possible interactions as it may include sequences that are not executable. Formally, $S$ may be a strict subset of $T$. This flexibility enables users to ignore low-level design details, such as system initialization or invalid calls in certain configurations.

The test space, represented as a single automaton, can quickly become very large, even based on simple specifications, as we will show in \autoref{sec:example}. Bar Sinai et al.~\cite{bar-sinai_provengo_2023} proposed more effective methods for describing this space. Since this paper focuses on extracting bug descriptions from the failing tests, we do not elaborate on breaking the test model into components here. In practice, the two methods can and should be combined. Furthermore, we demonstrate in \autoref{sec:evaluation} that our method does not require the explicit construction of $T$.

%

\section{A Running Example: Testing an Online Store}
\label{sec:example}

Our running example is based on an experiment conducted using a real online store system, with some simplifications for clarity and conciseness~\cite{bar-sinai_provengo_2023}
Testers interact with the system through both the REST API and the UI, using tools such as Selenium~\cite{8117878} or Playwright~\cite{10714759} to simulate user behavior. The key test actions include `\StartSession', `\Login', `\Logout', `\AddToCart', `\RemoveFromCart', and `\Checkout'. Each action takes parameters that may reference objects, such as product IDs or user credentials, created in previous steps.

A typical test flow involves multiple concurrent sessions, with users potentially logging into multiple sessions simultaneously. Within each session, users can add or remove products from their cart and attempt to check out. The test suite verifies both positive scenarios (e.g., successful checkout with valid cart contents) and negative ones (attempting operations that should fail, e.g., removing non-existent items or checking out with an empty cart). Each test validates whether the system correctly accepts or rejects the operation and checks for data consistency, such as ensuring that the checkout total reflects all cart modifications.


The model we use in this paper is based on the automaton shown in \autoref{fig:test-space} with the alphabet:
\begin{align*}
\Sigma = \bigcup_{i=1}^{N} \bigcup_{u \in U} \bigcup_{p \in P}  \{ &\StartSession_i, \Login_i(u),  \Logout_i \\ & \AddToCart_i(p), \RemoveFromCart_i(p), \Checkout_i\},
\end{align*}
where $U$ is a set of users and $P$ is a set of products.

This automaton represents the main flow of an online session $i$, with transitions corresponding to the basic shopping actions. 
$u$ and $p$ denote arbitrary values drawn from pre-determined finite sets during test execution. 
Multiple such automata run in parallel, following a standard composition approach where at each step, one automaton transitions while the others remain in their current states. A run is accepted if at least one automaton reaches an accepting state. 

For example, the following test simulates two sessions interacting simultaneously with the system:
$$S_1 L_1 A_1 R_1 A_1 S_2 L_2 C_2 C_1 \in A(1) || A(2)$$
where
$S_i$ denotes $\StartSession_i$;
$L_i$ denotes $\Login_i(u)$ for some user $u$;
$A_i$ and $R_i$ denote $\AddToCart_i(p)$ and $\RemoveFromCart_i(p)$ for some product $p$, respectively;
and $C_i$ denotes $\Checkout_i$.
Here, $\mathcal{A}(i)$ represents the automaton for the $i$th session. This test demonstrates a scenario where a user logs into two sessions simultaneously. 

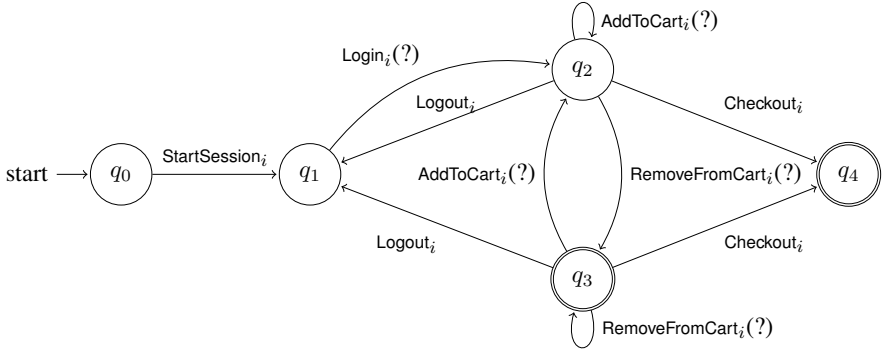
\begin{figure*}[h]
  \begin{minipage}[c]{0.53\textwidth}
    \centering
    \scalebox{0.93}{
        \begin{tikzpicture}[shorten >=1pt, node distance=3cm, on grid, auto] 
           \node[state,initial] (q_0)   {$q_0$}; 
           \node[state] (q_1) [right= 2.7cm of q_0] {$q_1$}; 
           \node[state] (q_2) [above right=1.5cm and 3.9cm of q_1] {$q_2$}; 
           \node[state,accepting] (q_3) [below right=1.5cm and 3.9cm of q_1] {$q_{3}$}; 
           \node[state,accepting] (q_4) [right=7.7cm of q_1] {$q_4$}; 
           \node[draw=none, fill=none, inner sep=0pt] (txtq2) [right=1.1cm of q_2, yshift=0.7cm] {$\AddToCart_i$(?)};
           \node[draw=none, fill=none, inner sep=0pt] (txtq3) [right=1.5cm of q_3, yshift=-0.7cm] {$\RemoveFromCart_i$(?)};
           
           \path[->] 
            (q_0) edge node {$\StartSession_i$} (q_1)
            (q_1) edge[bend left] node {$\Login_i$(?)} (q_2)
            (q_2) edge node[above] {$\Logout_i$} (q_1)
            (q_3) edge node {$\Logout_i$} (q_1) 
            (q_2) edge[loop above] node {} ()
            (q_2) edge[bend left] node {$\RemoveFromCart_i$(?)} (q_3)
            (q_3) edge[loop below] node {} ()
            (q_3) edge[bend left] node {$\AddToCart_i$(?)} (q_2)
    
            (q_2) edge node {$\Checkout_i$} (q_4)
            (q_3) edge node[below right] {$\Checkout_i$} (q_4);
        \end{tikzpicture}
     }
  \end{minipage}\hfill
  \begin{minipage}[c]{0.32\textwidth}
    \caption{An automaton representing the main flow of an online session $i$.
    The question marks denote arbitrary parameters: user $u \in U$ for the login and product $p \in P$ for the add and remove, where $P$ and $U$ are finite sets. 
    Multiple such automata run in parallel, following a standard composition approach in which, at each step, one automaton transitions while the others remain in their current states. A run is accepted if at least one of the automata is in an accepting state.}
    \label{fig:test-space}
  \end{minipage}
\end{figure*}

This relatively simple domain already demonstrates the complexity of test generation. The interactions between multiple sessions, stateful operations, and data dependencies create a rich space for potential bugs. Model-based testing is well-established for such complex systems~\cite{utting2012practical}, where finite automata describe possible action sequences. While these models are traditionally used to generate tests and measure coverage~\cite{lee1996principles,tretmans2008model}, we focus on using them to extract concise bug descriptions. We assume that testers have an automaton or equivalent regular language representation from which tests are extracted.

While our approach models tests as words over a finite alphabet, testing typically involves interacting with a system and observing its responses. This leads to a natural question: How do we represent the system's expected reactions? In our framework, the model encodes only high-level testing instructions, which are then interpreted by a testing infrastructure that mediates system interaction. This infrastructure executes commands and classifies test results as passed, failed, or invalid based on system responses. For example, when the infrastructure receives the word $S_1 L_1 A_1 R_1 A_1 S_2 L_2 C_2 C_1$, it executes each action sequentially, validates responses after each step, checks specific assertions like empty cart errors, and classifies the test outcome.
If expected behaviors are not observed, such as an error message appearing at the wrong time or not appearing when it should, the test is marked as failed. This separation of concerns allows us to focus on test generation and bug detection patterns while delegating response validation to the testing framework. The framework handles the complexities of system interaction and assertion checking, letting us concentrate on identifying and explaining bug patterns.

\section{Failure Explanation}
\label{sec:fe}
Since we aim to learn the language of failing tests, a natural approach is to learn the language $B$. However, while $B$ fully characterizes which tests fail, it often provides limited value to developers. The reason is that $B$ often includes considerable superfluous information that obscures the core reason for failure, as the following example demonstrates:

\begin{example} 
\label{ex:B}
Consider a simple bug in the online store where every time we try to remove a certain product, $x$, from the cart after adding it twice, the system crashes. A simple description of this bug would be a language $E$ defined by $\Sigma^* \AddToCart_{?}(x) \Sigma^* \AddToCart_{?}(x) \Sigma^* \RemoveFromCart_{?}(x)$. However, the language $B = E \cap S$ contains additional constraints from $S$, such as proper session initialization, login requirements, and checkout validation rules. These constraints, while important for valid test sequences, obscure the core reason for failure, the double-add followed by a remove pattern. Minimal DFAs for $E$ and $B$ are given in Appendix A in~\cite{supplementarymaterial}.
\end{example}

Therefore, rather than attempting to learn $B$ exactly, we seek an approach that can factor out irrelevant details while preserving the essential patterns that distinguish failing tests from passing ones. Simply specifying the failing tests $B$ as the target language to be learned is too restrictive --- it forces us to include all the constraints from $T$ even when they are irrelevant to understanding the bug.

Instead, we propose a more flexible formulation by separating what must be accepted (failing tests, $B$) from what must not be accepted (passing tests, $S\setminus B$). This separation gives us the freedom to abstract away irrelevant constraints while still correctly classifying all test cases. The remaining sequences in $\Sigma^* \setminus S$ can be treated as ``don't care'' cases.

To formalize our approach, we first introduce the notion of constraint specifications~\cite{chen_learning_2009}:
\begin{definition}
\label{def:consistency}
A \emph{constraint specification} is a pair $(A,R)$ where $A,R\subseteq \Sigma^*$ and $A\cap R=\emptyset$. A language $L$ is said to be \emph{consistent with $(A,R)$} if $A\subseteq L$ and $L\cap R = \emptyset$.
\end{definition}

Intuitively, $A$ represents words that must be accepted (positive examples), and $R$ represents those that must not be accepted (negative examples). Any sequence in $\Sigma^* \setminus (A\cup R)$ is considered a ``don't care'' word --- it may be either accepted or rejected based on what leads to a simpler bug description.

\begin{definition}
\label{def:fe}
$L$ is called a \emph{failure explanation} (FE) for $S$ and $B$ if it is consistent with $(B, S\setminus B)$.
\end{definition}

This definition formalizes our approach through three complementary components. First, a constraint specification defines what sequences must be accepted or rejected, establishing the basic boundaries of the problem. Second, consistency ensures a language respects these constraints, providing a mechanism for validation. Third, a failure explanation specifically applies consistency to the constraint specification derived from the set of tests the system can
execute $S$ and bug set $B$, focusing the analysis on the relevant test outcomes. Together, this framework allows us to characterize bug patterns while abstracting away irrelevant details precisely. For illustration, in \autoref{ex:B}, $E$ is a failure explanation (FE) for $S$ and $B$. 

To formalize the benefits of using failure explanations, we next analyze the state complexity of automata that represent languages. 
Our choice of focusing on the size of the minimal DFA is motivated by the fact that smaller automata are generally easier to understand and analyze. 
For a language $L$, let $\mathcal{L}$ denote the minimal DFA (i.e., the DFA with the fewest states) that accepts $L$. We use $|\mathcal{L}|$ to represent the number of states in this minimal DFA.

The following lemma shows that a minimal failure explanation can only smaller than the minimal DFA for the bug set $B$:
\begin{lemma}
Let $\mathcal{B}$ be a minimal DFA accepting $B$ and $\mathcal{E}$ be a minimal DFA such that $E$ is a failure explanation, then $|\mathcal{E}| \leq |\mathcal{B}|$.
\end{lemma}
\begin{proof}
Since $B$ is also an FE for $T$ and $B$, any minimal FE DFA for $T$ and $B$ is smaller than the minimal DFA for $B$. 
\end{proof}

Note that explanations can meet \autoref{def:fe} yet remain unhelpful, much like a textbook that conveys accurate information in a way that readers find difficult to understand. The following example illustrates this:
\begin{example}
\label{ex:E}
Consider an automaton for the bug described in \autoref{ex:B} that recognizes all words containing more occurrences of $\Login_?(?)$ than $\Logout_?$ (with some bounded counter to maintain regularity) that also satisfy the pattern $\Sigma^* \AddToCart_{?}(x) \Sigma^* \AddToCart_{?}(x) \Sigma^* \RemoveFromCart_{?}(x)$. Although this automaton meets the criteria for an error explanation, accepting all words in $B$ and rejecting all words in $S \setminus B$, it is not a good explanation, as it imposes the above constraints that are unrelated to the actual bug. 

\end{example}

This example illustrates that a DFA satisfying the consistency requirements can be highly complex. However, as demonstrated in \autoref{ex:B}, the flexibility to accept words outside \( T \) enables us to pursue more general and meaningful explanations rather than merely learning \( B \).

\section{Eventual Failure Explanation}
\label{sec:efe}
Failure explanations focus on failing tests, but this strict definition can sometimes obscure the true source of a bug. The following example illustrates this point:

\begin{example}
\label{ex:efe}
Consider a bug in the online store case study: If product \( x \) is added and removed from the cart, an incorrect error message appears at checkout. 
By \autoref{def:fe}, the language \( \Sigma^* \AddToCart_?(x) \Sigma^* \RemoveFromCart_?(x) \) is not a failure explanation, as it must accept a word only after \( \Checkout \). Consequently, failure explanations fail to direct bug hunters to examine the faulty \( \RemoveFromCart \) action. 
\end{example}

This example suggests that limiting failure explanations to failing tests alone may be insufficient. Allowing certain passing tests to be marked as faulty can provide additional insights into the root cause of an assertion violation, ultimately leading to more effective debugging.

We propose allowing explanations to accept passing tests when they are bound to eventually fail. This perspective aligns with the previous example, where the faulty $\RemoveFromCart$ action is causing the failure, even if it does not immediately trigger an assertion violation. By recognizing such cases, we can provide more informative failure explanations that better guide debugging efforts.  

For simplicity, we first define the concept of \emph{may pass} and then introduce the main concept of \emph{eventually fail} as its negation, as follows:

\begin{definition}
\label{def:may-pass}
    A test prefix \( w \in \Pref(S) \) \emph{may pass} if one of the following holds:
    \begin{enumerate}
        \item There exists \( u \) such that \( wu \in S \setminus B \), and for all \( v \), \( wuv \notin S \). In words, \( w \) can be extended to a passing test that cannot be further extended into a longer test.
        \item There exists a sequence of \( u_1, u_2, \dots \), where for each \( i = 1, 2, \dots \), \( u_i \neq \varepsilon \) and \( wu_1u_2\cdots u_i \in S \setminus B \).  In words, there is an infinite continuation of $w$ that infinitely often is a passing test.
    \end{enumerate}
    We say that \( w \) \emph{will eventually fail} if \( w \) may not pass.
\end{definition}

Using the notion of an eventually failing test, we can now define eventual failure explanation (EFE) languages:

\begin{definition}
\label{def:eventual-failure-explanation}
    Let \( \lozenge B \) denote the set of all tests in \( T \) that will eventually fail.\footnote{\(\lozenge\) is the linear temporal logic (LTL) \emph{eventually} operator, and is thus ``borrowed" here to represent tests that will eventually fail.} A language \( L \) is an eventual failure explanation (EFE) if it is consistent with \( (B, T \setminus \lozenge B) \).
\end{definition}

Comparing the above definition with \autoref{def:fe}, we observe that the eventual failure explanation (EFE) is looser. While a failure explanation (FE) language must reject all passing tests, an EFE offers more flexibility: it may accept a passing test if it will eventually fail. Therefore, FE emphasizes \emph{which} tests will fail, as it accepts only failing tests and non-tests. The looser language EFE emphasizes \emph{when} tests will fail, as it may accept a passing test $w$ if $w\in \lozenge B$. \autoref{ex:efe} illustrates this observation. We remark that EFE is \emph{not} a design choice to replace FE. Engineers can use both to gain complementary insights about the bugs they hunt.

Another perspective on EFE can be provided by the \emph{Reachability, Infection, Propagation, and Revealability} (RIPR) model of software testing~\cite{ammann2017introduction}. In this context, an EFE exploits situations where propagation to the fault-revealing location is guaranteed.

\section{Minimal Extension Closed Explanation}

The challenge is to select a failure or eventual failure explanation from the many languages satisfying \autoref{def:fe} and \autoref{def:eventual-failure-explanation}. A natural choice is the explanation with the smallest automaton. Also, we prefer extension closed explanations as they simplify reasoning by ensuring failures remain failures under extension, preventing transitions from accepting to non-accepting states.

Recall that we assumed \( B \) is extension closed with respect to \( S \), as failures from assertion violations persist under extension (\autoref{ass:test-model}).  
The following proposition shows that this assumption ensures that extension closedness can be enforced without compromising minimality. In other words, we can construct the smallest possible automaton recognizing an extension closed FE/EFE language while preserving this property.

\begin{proposition} 
If \( B \) is extension closed with respect to \( S \) (i.e., if \( w \in B \) and \( wu \in S \), then \( wu \in B \)), then there exists an extension closed FE/EFE language whose minimal automaton is the smallest among all automata recognizing FE/EFE languages, respectively. 
\end{proposition}
\begin{proof}[Proof sketch]
Let \(\mathcal{A}\) be a minimal FE DFA for  some \( S \) and \( B \). We construct a new DFA by modifying \(\mathcal{A}\) such that every transition from an accepting state transitions back to itself. This transformation ensures that the resulting DFA is extension closed. The resulting DFA still satisfies the FE consistency condition defined in \autoref{def:fe}, since by assumption, \( B \) is extension closed with respect to \( S \): if a word \( w \in B \) reaches an accepting state, then any extension $wu$, can only be in $B$ or not in $S$. Therefore, the modified DFA does not violate the consistency condition. This transformation does not increase the number of states, so the resulting DFA remains minimal. A similar argument applies for a minimal EFE DFA, using the consistency condition in \autoref{def:eventual-failure-explanation}.  
\end{proof}

This proposition is a practical heuristic that gives a criterion for selecting among all possible FE and EFE languages.

\section{Early Detection}
\label{sec:early-detection}


Since computing a minimal automaton is challenging (see \autoref{sec:discussion} for a brief account of our attempts), we propose restricting the search space by imposing additional constraints on the language. The following example illustrates the use of such constraints:

\begin{example}  
Consider a bug that manifests as an unexpected error message during product removal or checkout, occurring when a user logs out and then logs back in after adding a product to the cart. A good failure explanation is given by the language \(\Sigma^* \AddToCart_i(?) \Sigma^* \Logout_i \Sigma^*\), which serves as both an FE and an EFE. Its key advantage over other FEs and EFEs is that it directs testers to the earliest point where the failure can be detected, facilitating quicker diagnosis and debugging. Using only EFE does not necessarily lead to this language, as it leaves the classification open for words not in $S$.
\end{example}

This example demonstrates the importance of marking test prefixes that are destined to fail - when all their potential extensions into valid tests will also fail. This illustrates early detection in the context of failure explanations, and the principle also applies to eventual failure explanations. The key insight is that early detection should pinpoint problematic behavioral patterns at the earliest possible stage, before they result in actual failures. We formalize this concept for arbitrary constraint specifications $(A,R)$ by defining an early detection language as one that accepts words as soon as all their valid extensions intersect $A$ while remaining disjoint from $R$.

\begin{definition}
\label{def:early-detection}
    A language $L$ that is consistent with a language-requirements $(A,R)$ is called \emph{early detection} (ED), if for every $w$, if $w\in \Pref(A)$ and $w\notin \Pref(R)$, then $w\in L$.
\end{definition}

There is an equivalent formulation of early detection that provides deeper insight into its purpose: A language $L$ consistent with $(A,R)$ is early detection (ED) if and only if it is consistent with $(\text{Pref}(A)\setminus\text{Pref}(R), R)$. This alternative formulation reveals two key aspects of early detection:
\begin{enumerate*}[label=\arabic*)]
    \item By accepting words in $\Pref(A)\setminus\Pref(R)$, we identify problematic patterns as soon as they appear in the execution trace, before they manifest as actual failures.
    \item By maintaining consistency with $R$, we ensure that the explanation remains sound - it never misclassifies known passing tests.
\end{enumerate*}

In the remainder of the paper we use EDFE and EDEFE as abbreviations for early detection failure explanation and early detection eventual failure explanation, respectively. These concepts combine the benefits of early detection with our previous notions of failure explanations, allowing us to both identify problematic patterns early and maintain the appropriate level of abstraction.

A natural question that arises is whether the early detection constraint preserves minimality, as in the case of extension closeness. We show that this is not always true: there exist instances where a minimal FE (EFE) DFA has strictly fewer states than the corresponding minimal EDFE (EDEFE) DFA.

\begin{proposition}
A minimal FE (EFE) DFA can be strictly smaller than an EDFE (EDEFE) DFA. 
\end{proposition}
\begin{proof}
For $\Sigma = \{ 0,1 \}$, $T = \Sigma^*$, 
    $S = 00\Sigma^* + 1\Sigma^*$, 
    and $B = 00\Sigma^* + 1(1\Sigma)^*0\Sigma^*$:\phantom{\qedhere}

{\small Minimal FE \& EFE:}
\scalebox{0.7}{
\begin{tikzpicture}[baseline=-0.85ex,shorten >=1pt,node distance=3cm,on grid,auto, initial text=] 
    \node[state,initial] (s) {$s$}; 
    \node[state] (q) [right=of s] {$q$}; 
    \node[state,accepting] (r) [right=2cm of q] {$r$}; 
    \path[->] 
    (s) 
        edge [bend left] node {0, 1} (q)
    (q) edge [bend left] node {1} (s)
        edge [] node {0} (r)
    (r) edge [loop right] node {0, 1} (r);
\end{tikzpicture}
}

{\small Minimal EDFE \& EDEFE:}
\scalebox{0.7}{
\begin{tikzpicture}[baseline=-0.85ex, shorten >=1pt,node distance=3cm,on grid,auto, initial text=] 
    \node[state,initial] (s) {$s$}; 
    \node[state] (q1) [below=1.7cm of s] {$q_1$}; 
    \node[state] (q2) [right=of q1] {$q_2$}; 
    \node[state,accepting] (r) [right=of s] {$r$}; 
    \node[draw=none, fill=none, inner sep=0pt] (dummy) [right=2.0cm of q2] {\qed};
    \path[->] 
    (s) edge [] node {0} (r)

        edge [] node {1} (q1)
    (q1) edge [] node {0} (r)
         edge [] node {1} (q2)
    (q2) edge [bend left] node {0, 1} (q1)
    (r) edge [loop right] node {0, 1} (r);
\end{tikzpicture} 
}

\end{proof}

We remark that by enlarging the alphabet size, for all $n\in \mathbb{N}$, we can extend this example into one in which the difference in the number of states is greater than $n$.
Our experiments, detailed later in \autoref{sec:evaluation}, show no statistically significant advantage of early detection in general. However, we observed specific cases where early detection resulted in a significantly more compact representation.


\section{Algorithms for Finding Error Explanations and Eventual Failure Explanations}
\label{sec:algorithm}

To find error explanations, we propose an algorithm that we elaborate on in this section. The first stage involves finding a \emph{3-valued DFA (3DFA)} with an $L^*$ algorithm using the concepts presented in~\cite{leucker_learning_2012}. A 3DFA $\mathcal{A} = \tuple{Q,\Sigma,q_0,\delta,Acc,Rej,Dont}$ is a finite automaton whose states are partitioned into three labels: accepting, rejecting, and ``don't care''. In our context, we look for a 3DFA to represent the relevant information on our investigated SUT:

\begin{definition}
Let $\mathcal{A} = \tuple{Q,\Sigma,q_0,\delta,Acc,Rej,Dont}$ be a 3DFA and let
$\mathcal{A}^{C}=\tuple{Q,\Sigma,q_0,\delta,C}$ be a DFA for each $C \in \{ Acc,Rej,Dont\}$. 
Let $T$, $S$, and $B$ be as in \autoref{ass:test-model}.
We say that a 3DFA \emph{captures} these languages if $L(\mathcal{A}^{Acc}) = B$ and $L(\mathcal{A}^{Rej}) = S \setminus B$.
\end{definition}



{\bfseries L* for 3DFA.} To obtain a 3DFA that captures our system, we propose a teacher for the $L^*$ algorithm. We assume our teacher has access to the test model $T$, and to an optional initial test repository consisting of tests for which the system output is already known. This approach facilitates a scenario where the developer already has a dataset of test cases that initially led to the discovery of the bug. The teacher will initially attempt to respond to queries using the model $T$ and the existing tests, in order to reduce the number of real executions on the SUT, which can be prohibitively expensive.

For membership queries, given $w = \Sigma^*$, the teacher first checks if $w$ is already in the test repository and if it is, responds accordingly. Otherwise, if $w \notin T$ the teacher answers $Dont$. Finally, if $w \in T$, the teacher runs $w$ on the SUT and answers $Acc$ if $w \in B$, $Dont$ if $w \notin S$, and $Rej$ otherwise.

The equivalence query process is outlined in \autoref{alg:equivalence}. To avoid unnecessary SUT executions, the algorithm first checks if a counterexample for the candidate 3DFA $\mathcal{A}$ already exists in the test repository. If not, the algorithm checks if $L(\mathcal{A}^{Acc \cup Rej}) \subseteq T$. If not, the teacher returns a counterexample $w \in L(\mathcal{A}^{Acc \cup Rej}) \setminus T$. Otherwise, the algorithm continues and checks if: 
\begin{enumerate*}[label=\arabic*)]
    \item $L(\mathcal{A}^{Acc}) = B$, 
    \item $L(\mathcal{A}^{Rej}) = S \setminus B$, and 
    \item $L(\mathcal{A}^{Dont}) = \Sigma^* \setminus S$.
\end{enumerate*}

Since we do not know $S$ and $B$ in advance, these checks are done heuristically through random interactions with the SUT~\cite{1702519,533956}. If a counterexample is found, the process finishes in any of the three steps. The reason we make three separate checks instead of one that interacts directly with $\mathcal{A}$ is to emphasize each of the three classes and reduce the effect of underrepresented classes. For example, finding a random counterexample from the class of bugs $Acc$, which is usually rare, is harder in $\mathcal{A}$ than in $\mathcal{A}^{Acc}$. Also, these three checks are independent and can be executed in parallel until the first counterexample is found. Note that the computation of $L(\mathcal{A}^{Acc \cup Rej}) \subseteq T$ can also be done heuristically. This approach provides more flexibility as it allows for the definition of a non-regular test-space $T$ or for reducing the computational complexity in case $T$ is large, as discussed in \autoref{sec:evaluation}.

\begin{algorithm}
\caption{Equivalence Query for Candidate 3DFA $A$}
\label{alg:equivalence}
\renewcommand{\algorithmicrequire}{\textbf{Input:}}
\begin{algorithmic}[1]
    \Require{$\mathcal{A}=\tuple{Q,\Sigma,q_0,\delta,Acc,Rej,Dont}$, $T$}
    \State{$cex \gets \text{get\_cex\_test\_repo}(\mathcal{A})$}
    \If{$cex \ne null$}
        \Return{$cex$} \EndIf
    \If{$L(\mathcal{A}^{Acc \cup Rej}) \not \subseteq T$}
        \Return{$w \text{ such that } w \in L(\mathcal{A}^{Acc \cup Rej}) \setminus T$} \EndIf
    \State{$cex \gets \text{get\_cex\_heuristically}(\mathcal{A}^{Acc})$}
    \If{$cex \ne null$}
        \Return{$cex$} \EndIf
    \State{$cex \gets \text{get\_cex\_heuristically}(\mathcal{A}^{Rej})$}
    \If{$cex \ne null$}
        \Return{$cex$} \EndIf
    \State{$cex \gets \text{get\_cex\_heuristically}(\mathcal{A}^{Dont})$}
    \If{$cex \ne null$}
        \Return{$cex$} \EndIf
    \Return{$null$}
\end{algorithmic}
\end{algorithm}

Once a 3DFA is computed, the algorithm proceeds to the second stage to compute an error explanation DFA. We assume that the 3DFA is minimal, because this is guaranteed by $L^*$. Depending on the type of explanation, the 3DFA might go through a relabeling process before the DFA computation, as described in the following paragraphs. 

{\bfseries 3DFA manipulation for EFE.} The relabeling process is presented in \autoref{alg:eventual-fe} for an eventual failure explanation (EFE). It starts by computing which states may pass based on the two conditions in \autoref{def:may-pass}. For the first condition, it computes which states can reach a $Rej$ state that transitions immediately to a sink $Dont$ state (lines 1-3). Since the 3DFA is minimal, there is only one such state. For the second condition, the algorithm finds strongly connected components in the 3DFA graph and then computes which states can reach a component that contains a $Rej$ state (lines 4-10). Finally (lines 11-13), it updates the 3DFA to reflect the consistency condition described in \autoref{def:eventual-failure-explanation} by adding the found states to $Dont$ and removing them from $Rej$.

\begin{algorithm}
\caption{Eventual FE Relabeling for 3DFA $A$}
\label{alg:eventual-fe}
\renewcommand{\algorithmicrequire}{\textbf{Input:}}
\begin{algorithmic}[1]
\Require{$A = \tuple{Q,\Sigma,q_0,\delta,Acc,Rej,Dont}$}
    \State{$sinkDont \gets \{ q \in Dont  \mid \forall \sigma \in \Sigma  (\delta(q,\sigma) = q) \}$}
    \State{$mayPass \gets \{ q \in Rej \mid \forall \sigma \in \Sigma (\delta(q,\sigma) \in sinkDont) \}$}
    \State{$mayPass \gets \text{backward\_reachability}(A, mayPass)$}
    \State{$components \gets \text{get\_strongly\_connected\_components}(A)$}
    \For{$C$ in $components$}
        \If{$C \cap Rej \ne \emptyset$}
            \State{$mayPass \gets mayPass \cup C $}
        \EndIf
    \EndFor
    \State{$mayPass \gets \text{backward\_reachability}(A, mayPass)$}
    \State{$Rej' \gets Rej \cap mayPass$}
    \State{$Dont' \gets Dont  \cup (Rej \setminus Rej')$}
    \State{\Return{$\tuple{Q,\Sigma,q_0,\delta,Acc,Rej', Dont'}$}}
\end{algorithmic}
\end{algorithm}

{\bfseries 3DFA manipulation for ED.} For early detection explanation (ED), we relabel the 3DFA states based on \autoref{def:early-detection} using the process shown in \autoref{alg:early-detection}: First, the sets of states that can reach $Acc$ and $Rej$ (lines 1-2) are computed. Then (lines 3-5), the 3DFA is updated to reflect the consistency described in \autoref{def:early-detection} by adding to $Acc$ states from which $Acc$ is reachable, but $Rej$ is not reachable.

\begin{algorithm}
\caption{Early Detection Relabeling for 3DFA $A$}
\label{alg:early-detection}
\renewcommand{\algorithmicrequire}{\textbf{Input:}}
\begin{algorithmic}[1]
\Require{$A = \tuple{Q,\Sigma,q_0,\delta,Acc,Rej,Dont}$}
    \State{$reachAcc \gets \text{backward\_reachability}(A, Acc)$}
    \State{$reachRej \gets \text{backward\_reachability}(A, Rej)$}
    \State{$Acc' \gets reachAcc \setminus reachRej$}
    \State{$Dont' \gets Dont \setminus Acc'$}
    \State{\Return{$\tuple{Q,\Sigma,q_0,\delta,Acc',Rej,Dont'}$}}
\end{algorithmic}
\end{algorithm}

{\bfseries 3DFA manipulation for EDEFE.} For early detection eventual failure explanation (EDEFE), the algorithm first relabels the 3DFA using the process described in \autoref{alg:eventual-fe}. Then, it proceeds with the updated 3DFA to \autoref{alg:early-detection}.

{\bfseries 3DFA to 2DFA.} The final stage of the algorithm involves computing the final DFA from the 3DFA. There are various possible approaches to this step, with varying levels of computational complexity, for example, the one proposed by Paull and Unger \cite{paull_minimizing_1959}. We rely on a passive learning algorithm to minimize the runtime overhead, but this choice is not inherent to the method we present. Further details and implementation details are provided in the following section.




\section{Evaluation}
\label{sec:evaluation}


We continue with a quantitative evaluation of the proposed approaches. 
Our guiding research questions are:
\begin{enumerate}[label=\textbf{RQ\arabic*.}, ref=\textbf{RQ\arabic*}, leftmargin=*]
    \item \label{itm:rq1} Can FE/EFE automata produce smaller descriptions than B?  
    \item \label{itm:rq2} Do ED automata improve description size?
    \item \label{itm:rq3} Do our approaches apply to realistic systems?
\end{enumerate}
To this end, the algorithms described in \autoref{sec:algorithm} were implemented using the AALpy automata learning library~\cite{muskardin_aalpy_2022}. For the heuristic equivalence query answering, we used a randomized version of the W-method algorithm~\cite{1702519} implemented in AALpy. To compute the final DFA from the 3DFA, we implemented a method that derives a sample of words from the 3DFA, which either reach $Acc$ states or $Rej$ states, denoted as $S = (S^{Acc}, S^{Rej})$. Subsequently, it identifies a minimal DFA consistent with $S$ using the RPNI algorithm~\cite{oncina_identifying_1992} implemented in AALpy. The complete code for the evaluations and the tool we have implemented are publicly available at~\cite{supplementarymaterial}. Experiments were conducted on an Xeon E5-2620 CPU with 32GB RAM.

\subsection{The RERS Challenge 2019 Benchmark}

The Rigorous Examination of Reactive Systems (RERS) 2019~\cite{jasper_rers_2019} contains several industrial benchmarks. 
These benchmarks are control software components of ASML's TWINSCAN lithography machines. The software components are given in Java code, and are categorized into two types: arithmetic computation and data structures.

Each software component functions as a reactive system that processes user input at each time step, updates its internal state variables accordingly, and generates an output in response. Some executions may result in assertion violations, leading the system to fail with a specific error code. Additionally, not all inputs are valid at all states, and the system terminates if such an invalid input is received at an incompatible state in a detectable way, which we refer to as the INVALID errors. 

We derived our evaluation dataset from the RERS 2019 benchmark as follows:
we began by executing the Java Pathfinder model checker~\cite{havelund_model_2000} on each system and each of its corresponding error codes, ignoring traces leading to INVALID errors.
This yielded a set of system traces for a subset of error codes for each system.
We filtered this set to a single trace, corresponding to a single system-error code pairing, of length 5 to 15. This pre-processing step resulted in traces for 27 out of 30 systems, with the remaining three excluded due to compilation issues or missing/short counterexamples. Our evaluation used a single error code for each system to ensure a focused bug description and to assess a more realistic scenario. 

We concentrated our analysis on alphabets of sizes 10 to 25. For systems with larger alphabets, we selectively trimmed transitions of certain letters irrelevant to the counterexample we aimed to explain. This step is a natural approach when dealing with real-world systems, where most operations are not directly related to the bug that engineers are trying to address. Automation of this step and other ways to handle extremely large alphabets are left for future research.

In this section we refer to the language of all valid system traces as $S$. That means that given a user-defined test space $T$, we treat all words outside of $S \cap T$ as ``don't-care''.

\subsection{Simulating Different Types of SUTs}

To evaluate the effect of our approach on different types of SUTs, we constructed several test spaces ($T$) and system configurations ($S$) using the failing tests identified by the model checker. Given $\text{cex}_i$, the counterexample (failing test) for system $i$, we defined three setups:


    \textbf{Unrestricted setup ($\UNR$):} In this setting, $S_{\UNR,i} = S_i$ and $T_{\UNR,i}=\Sigma_i^*$. This means that we do not restrict the test space, and any run allowed by the system ($S_i$) is considered a valid test.
    
    \textbf{Failure-driven setup ($\FDR$):} Here we also maintain the original system configuration, i.e., $S_{\FDR,i} = S_i$, but we define $T_{\FDR,i}$ to accept only words that contain all letters of $\text{cex}_i$.

    \textbf{Assertion-driven setup ($\ADR$):} In this setup, we updated the system configuration $S_i$ with a wrapper $S_{\ADR,i}$ that encapsulates $S_i$. The wrapper adds a letter, $\mathit{assert}$, to the alphabet. Moreover, if an error is encountered in $S_i$, $S_{\ADR,i}$ will only trigger the error after three $\mathit{assert}$. The goal here is to simulate a scenario where a system bug results in an assertion violation with a delay. We set $T_{\ADR, i}$ to accept words that end with $\mathit{assert}$, i.e., $\Sigma_i^*\mathit{assert}$. 
    

All tested benchmark systems and setups automata are available at~\cite{supplementarymaterial}.

We intentionally selected these setups to emphasize the significance of each failure explanation discussed in this paper. However, we believe they are general and cover a range of common system and failure behaviors, making them applicable to a broad spectrum of testing scenarios. By varying the test space and system configuration, we aim to capture different ways in which failures manifest and are detected in real-world systems. The $\UNR$ setup serves as a baseline, allowing unrestricted exploration of the system behaviors. The $\FDR$ setup focuses on elements contributing to a known failure. Finally, the $\ADR$ setup introduces a delayed error detection, modeling systems where failures propagate over time.

Our expectation was that $B$ would be bigger than our failure explanation in all setups. Specifically, in the $\UNR$ setup, we expected FE to be smaller than B. In the $\FDR$ setup, we expected the EDFE to be smaller than FE and B. In the $\ADR$ setup, we expected EFE to be smaller than FE and B.

\subsection{Results}

\begin{table*}[t]
    \centering
    \scriptsize
    \setlength{\tabcolsep}{1pt} 
    \begin{minipage}[t]{0.27\textwidth}
        \centering
        \begin{tabular}{c c c c c c c c c c}%
    & & 	& 	&  	& \bfseries 	Time 	& & 
    \\
    & \bfseries System	& \bfseries $|$cex$|$	& \bfseries $|\Sigma|$		& \bfseries 	$|$3DFA$|$ 	& \bfseries 	(sec.)  	& \bfseries $|$FE$|$ & \bfseries $|$B$|$ 
    \\\hline \hline
    \multirow{15}{*}{\rotatebox{90}{\bfseries Arithmetic Systems}} & m24 & 10 & 22 & 24 & 188 &\bfseries 4 & 12 \\ \cline{2-8}
    & m45 & 5 & 12 & 12 & 140 &\bfseries 4 & 8 \\ \cline{2-8}
    & m54 & 5 & 10 & 31 & 1721 &\bfseries 5 & 20 \\ \cline{2-8}
    & m55 & 14 & 25 & 33 & 1257 &\bfseries 7 & 26 \\ \cline{2-8}
    & m76 & 8 & 17 & 25 & 971 &\bfseries 4 & 19 \\ \cline{2-8}
    & m95 & 5 & 15 & 13 & 118 &\bfseries 4 & 7 \\ \cline{2-8}
    & m135 & 13 & 23 & 24 & 275 &\bfseries 4 & 18 \\ \cline{2-8}
    & m158 & 5 & 17 & 12 & 50 &\bfseries 4 & 7 \\ \cline{2-8}
    & m159 & 6 & 22 & 16 & 126 &\bfseries 4 & 8 \\ \cline{2-8}
    & m164 & 6 & 18 & 10 & 47 &\bfseries 2 & 8 \\ \cline{2-8}
    & m172 & 10 & 19 & 19 & 136 &\bfseries 4 & 12 \\ \cline{2-8}
    & m181 & 6 & 20 & 29 & 23556 &\bfseries 4 & 25 \\ \cline{2-8}
    & m183 & 5 & 12 & 12 & 251 &\bfseries 4 & 9 \\ \cline{2-8}
    & m185 & 15 & 23 & 47 & 995 &\bfseries 4 & 37 \\ \cline{2-8}
    & m201 & 7 & 19 & 23 & 688 &\bfseries 4 & 10 \\ \hline \hline
    \multirow{12}{*}{\rotatebox{90}{\bfseries \bfseries Data Structures}} & m22 & 9 & 23 & 26 & 394 &\bfseries 5 & 13 \\ \cline{2-8}
    & m27 & 9 & 25 & 15 & 116 &\bfseries 4 & 11 \\ \cline{2-8}
    & m41 & 5 & 22 & 12 & 274 &\bfseries 4 & 9 \\ \cline{2-8}
    & m106 & 11 & 24 & 22 & 330 &\bfseries 4 & 18 \\ \cline{2-8}
    & m131 & 6 & 21 & 10 & 120 &\bfseries 2 & 8 \\ \cline{2-8}
    & m132 & 5 & 20 & 12 & 13092 &\bfseries 3 & 8 \\ \cline{2-8}
    & m167 & 8 & 25 & 15 & 141 &\bfseries 3 & 10 \\ \cline{2-8}
    & m173 & 6 & 21 & 22 & 1695 &\bfseries 4 & 17 \\ \cline{2-8}
    & m182 & 6 & 14 & 20 & 1455 &\bfseries 4 & 10 \\ \cline{2-8}
    & m189 & 7 & 21 & 13 & 88 &\bfseries 3 & 9 \\ \cline{2-8}
    & m196 & 7 & 24 & 75 & 25430 &\bfseries 9 & 36 \\ \cline{2-8}
    & m199 & 7 & 20 & 13 & 159 &\bfseries 4 & 11 \\ \hline \hline
    \end{tabular}
        \captionof{table}{Results of the $\UNR$ setup}
     \label{tab:rers-results-1}
    \end{minipage}
    \hfill
    \begin{minipage}[t]{0.4\textwidth}
        \centering
        \begin{tabular}{c c c c c c c c c c}%
& & 	& & 	& 	& \bfseries 	Time  	&  & 	 	& 	
\\
& \bfseries System	& \bfseries $|$cex$|$	& \bfseries $|\Sigma|$	& \bfseries $|$T$|$	& \bfseries 	$|$3DFA$|$ 	& \bfseries 	(sec.)  	& \bfseries $|$FE$|$ & \bfseries $|$EDFE$|$	 	& \bfseries $|$B$|$ 	
\\\hline \hline
\multirow{15}{*}{\rotatebox{90}{\bfseries Arithmetic Systems}} & m24 & 10 & 22 & 512 & 21 & 142 &\bfseries 4 &\bfseries 4 & 12 \\ \cline{2-10}
& m45 & 5 & 12 & 32 & 9 & 100 & 4 &\bfseries 3 & 7 \\ \cline{2-10}
& m54 & 5 & 10 & 16 & 134 & 6094 &\bfseries 7 &\bfseries 7 & 20 \\ \cline{2-10}
& m55 & 14 & 25 & 1024* & \multicolumn{5}{c}{\emph{timeout}}\\ \cline{2-10}
& m76 & 8 & 17 & 256 & 268 & 886 &\bfseries 3 &\bfseries 3 & 10 \\ \cline{2-10}
& m95 & 5 & 15 & 32 & 39 & 92 &\bfseries 3 &\bfseries 3 & 7 \\ \cline{2-10}
& m135 & 13 & 23 & 1024* & 31 & 221 &\bfseries 4 &\bfseries 4 & 15 \\ \cline{2-10}
& m158 & 5 & 17 & 32 & 33 & 51 & 6 &\bfseries 4 & 7 \\ \cline{2-10}
& m159 & 6 & 22 & 64 & 48 & 111 & 8 &\bfseries 4 & 8 \\ \cline{2-10}
& m164 & 6 & 18 & 32 & 8 & 43 & 1 & 1 & 8 \\ \cline{2-10}
& m172 & 10 & 19 & 64 & 20 & 108 &\bfseries 4 &\bfseries 4 & 12 \\ \cline{2-10}
& m181 & 6 & 20 & 64 & 20 & 10016 &\bfseries 3 &\bfseries 3 & 8 \\ \cline{2-10}
& m183 & 5 & 12 & 32 & 60 & 602 & 25 &\bfseries 8 & 9 \\ \cline{2-10}
& m185 & 15 & 23 & & \multicolumn{5}{c}{\emph{out of memory}}\\ \cline{2-10}
& m201 & 7 & 19 & 64 & 25 & 280 &\bfseries 4 &\bfseries 4 & 9 \\ \hline  \hline 
\multirow{12}{*}{\rotatebox{90}{\bfseries \bfseries Data Structures}} & m22 & 9 & 23 & 512* & 56 & 390 &\bfseries 4 & 6 & 11 \\ \cline{2-10}
& m27 & 9 & 25 & 64 & 43 & 106 &\bfseries 6 &\bfseries 6 & 11 \\ \cline{2-10}
& m41 & 5 & 22 & 32 & 39 & 377 & 9 &\bfseries 5 & 11 \\ \cline{2-10}
& m106 & 11 & 24 & 1024 & 13 & 283 &\bfseries 3 &\bfseries 3 & 13 \\ \cline{2-10}
& m131 & 6 & 21 & 32 & 8 & 113 & 1 & 1 & 8 \\ \cline{2-10}
& m132 & 5 & 20 & 32 & 14 & 19405 &\bfseries 3 & 4 & 8 \\ \cline{2-10}
& m167 & 8 & 25 & 256 & 10 & 116 & 1 & 1 & 10 \\ \cline{2-10}
& m173 & 6 & 21 & 64 & 32 & 1015 &\bfseries 3 &\bfseries 3 & 8 \\ \cline{2-10}
& m182 & 6 & 14 & 32 & 12 & 1184 & 4 &\bfseries 3 & 9 \\ \cline{2-10}
& m189 & 7 & 21 & 128 & 17 & 83 &\bfseries 3 &\bfseries 3 & 9 \\ \cline{2-10}
& m196 & 7 & 24 & 64 & \multicolumn{5}{c}{\emph{timeout}}\\ \cline{2-10}
& m199 & 7 & 20 & 128* & 139 & 199 &\bfseries 4 &\bfseries 4 & 10 \\ \hline \hline
\end{tabular}
        \captionof{table}{Results of the $\FDR$ setup}
     \label{tab:rers-results-2}
    \end{minipage}
    \hfill
    \begin{minipage}[t]{0.32\textwidth}
        \centering
        \begin{tabular}{c c c c c c c c c c c}
& & & 	&  	& \bfseries 	Time  	& & 	& 	
\\
& \bfseries System	& \bfseries $|$cex$|$	& \bfseries $|\Sigma|$	&  \bfseries 	$|$3DFA$|$ 	& \bfseries 	(sec.)  	& \bfseries $|$FE$|$ & \bfseries $|$EFE$|$	 	& \bfseries $|$B$|$ 	
\\\hline  \hline
\multirow{15}{*}{\rotatebox{90}{\bfseries Arithmetic Systems}} & m24 & 13 & 23 & 42 & 383 & 6 &\bfseries 5 & 15 \\ \cline{2-9}
& m45 & 8 & 13 & 26 & 336 & 6 &\bfseries 5 & 12 \\ \cline{2-9}
& m54 & 8 & 11 & 68 & 5678 & 7 &\bfseries 4 & 25 \\ \cline{2-9}
& m55 & 17 & 26 & 84 & 3549 & 11 &\bfseries 8 & 32 \\ \cline{2-9}
& m76 & 11 & 18 & 38 & 1789 & 7 &\bfseries 4 & 17 \\ \cline{2-9}
& m95 & 8 & 16 & 28 & 265 & 7 &\bfseries 4 & 10 \\ \cline{2-9}
& m135 & 16 & 24 & 42 & 610 & 8 &\bfseries 5 & 18 \\ \cline{2-9}
& m158 & 8 & 18 & 26 & 128 & 7 &\bfseries 6 & 10 \\ \cline{2-9}
& m159 & 9 & 23 & 30 & 302 & 7 &\bfseries 4 & 11 \\ \cline{2-9}
& m164 & 9 & 19 & 20 & 127 & 5 &\bfseries 2 & 11 \\ \cline{2-9}
& m172 & 13 & 20 & 30 & 271 & 7 &\bfseries 4 & 15 \\ \cline{2-9}
& m181 & 9 & 21 & 36 & 17702 & 6 &\bfseries 5 & 20 \\ \cline{2-9}
& m183 & 8 & 13 & 28 & 718 & 8 &\bfseries 5 & 13 \\ \cline{2-9}
& m185 & 18 & 24 & 72 & 1628 & 7 &\bfseries 6 & 23 \\ \cline{2-9}
& m201 & 10 & 20 & 58 & 2640 & 9 &\bfseries 6 & 12 \\ \hline  \hline
\multirow{12}{*}{\rotatebox{90}{\bfseries \bfseries Data Structures}} & m22 & 12 & 24 & 64 & 935 & 6 &\bfseries 5 & 21 \\ \cline{2-9}
& m27 & 12 & 26 & 26 & 243 & 7 &\bfseries 4 & 14 \\ \cline{2-9}
& m41 & 8 & 23 & 28 & 770 & 7 &\bfseries 4 & 13 \\ \cline{2-9}
& m106 & 14 & 25 & 38 & 793 & 7 &\bfseries 4 & 20 \\ \cline{2-9}
& m131 & 9 & 22 & 24 & 241 & 5 &\bfseries 2 & 11 \\ \cline{2-9}
& m132 & 8 & 21 & 26 & 18806 & 6 &\bfseries 3 & 12 \\ \cline{2-9}
& m167 & 11 & 26 & 32 & 353 & 6 &\bfseries 3 & 13 \\ \cline{2-9}
& m173 & 9 & 22 & 32 & 2607 & 8 &\bfseries 5 & 15 \\ \cline{2-9}
& m182 & 9 & 15 & 40 & 2463 &\bfseries 8 &\bfseries 8 & 13 \\ \cline{2-9}
& m189 & 10 & 22 & 24 & 198 & 6 &\bfseries 5 & 12 \\ \cline{2-9}
& m196 & 7 & 24 & \multicolumn{5}{c}{\emph{timeout}} \\ \cline{2-9}
& m199 & 10 & 21 & 28 & 364 & 7 &\bfseries 4 & 14 \\ \hline  \hline
\end{tabular}
        \captionof{table}{Results of the $\ADR$ setup}
     \label{tab:rers-results-3}
    \end{minipage}
\end{table*}


We now present the evaluation results for the three setups described above. Each experiment is summarized in two tables: one for the arithmetic computations benchmark systems and the other for those involving data structures. The columns in the tables are as follows: ``System'' identifies a specific case study within the benchmark, $|$cex$|$ represents the length of the counterexample used, $|\Sigma|$ denotes the size of the alphabet, $|$3DFA$|$ indicates the size of the 3DFA obtained in the first stage of our algorithm, $|$FE$|$, $|$EFE$|$, and $|$EDFE$|$ correspond to the sizes of the automata we produce. These values are all compared against $|$B$|$, the size of the baseline automaton that describes the bug directly. The time presented in the tables is for the entire execution time of the algorithm presented in \autoref{sec:algorithm}. We note that the relabeling and RPNI steps were relatively negligible and did not take more than 5 seconds for each benchmark.

\autoref{tab:rers-results-1} presents the results of the $\UNR$ setup in our evaluation. A statistical analysis, using the Mann-Whitney (M-W) test, confirms that FE is significantly smaller than B (\( p \approx 10^{-10} \)), supporting the claim of its advantage. Furthermore, Spearman correlation tests show that the advantage (\( |\text{B}| - |\text{FE}| \)) is positively correlated with key parameters: it increases with longer counterexamples (\( \rho = 0.669, p \approx 10^{-4} \)), and larger alphabet sizes (\( \rho = 0.429, p \approx 0.026 \)), confirming the impact of these factors. 

Results of the $\FDR$ setup are presented in \autoref{tab:rers-results-2}. Here, we also include the size of $T$ derived from the counterexample, cex, and the size of the extracted EDFE. During execution, memory limitations affected four systems (marked with *), particularly when handling equivalence queries in the computation of $T$. These issues arose due to the large sizes of $T$ and the 3DFA, specifically during the subset computation (line 4) of \autoref{alg:equivalence}. We employed a heuristic alternative for these cases to mitigate this, applying the W-method directly over $T$ to identify counterexamples. Additionally, the evaluation timed out (exceeding 10 hours) for two systems, and one ran out of memory (exceeding 32GB).

The M-W test indicates that FE and EDFE are significantly smaller than $B$ ($p \approx 10^{-7}$ and $p \approx 10^{-9}$, respectively), while no significant difference was found between FE and EDFE ($p \approx 0.653$). Spearman correlation tests reveal a strong positive correlation between $B$ and the counterexample length ($\rho = 0.641, p \approx 10^{-4}$) and the test space size $|T|$ ($\rho = 0.498, p \approx 10^{-2}$). Neither FE nor EDFE significantly correlates with $|$cex$|$, $|\Sigma|$, or $|$T$|$.


Note that there are three cases where $|$FE$| = 1$. These correspond to systems in which $T$ is degenerated, meaning that $T$ became so focused that the only remaining tests were either the counterexample itself or highly similar failing executions.
Since such cases are unlikely to be useful for developers to describe the bugs in practice, we also conducted the statistical analysis excluding these instances.
The results remained consistent, showing similar correlations and significance.

Results of the $\ADR$ setup are presented in \autoref{tab:rers-results-3}. Here, we added the size of the EFE DFA found in the evaluation.  The statistical analysis of this setup confirms that both FE and EFE are significantly smaller than $B$, as demonstrated by the M–W tests ($p \approx 10^{-10}$ and $p \approx 10^{-10}$, respectively). A statistically significant difference was also observed between FE and EFE ($p \approx 10^{-7}$). Spearman correlation analysis reveals a strong positive correlation between B and the counterexample length ($\rho = 0.604, p \approx 10^{-3}$), as well as moderate correlations between both FE and EFE and the size of the 3DFA.

Recall that this experiment was conducted on a system where the bug consistently appears three assertions after its cause. In light of this, it is unsurprising that $|$EFE$|$ is often exactly three states smaller than $|$FE$|$. In most cases, the size reduction comes from omitting these extra steps, suggesting that engineers could use the larger FE automaton without being significantly affected by the additional three states.

\subsection{Discussion}
\label{sec:discussion}

Our evaluation shows that our approach can be applied to realistic systems, which answers \ref{itm:rq3}. Moreover, it demonstrates a statistically significant advantage of FE and EFE automata, addressing \ref{itm:rq1}. Regarding \ref{itm:rq2}, while EDFE automata improved the description size in some systems in the $\FDR$ setup, we observed no consistent advantage of ED. We believe that this is due to limitations in the test space definition, which may not align well with the systems under test. A more comprehensive evaluation is needed to fully assess the benefits of early detection for both FE and EFE.

In the process of implementing the algorithm discussed in \autoref{sec:algorithm}, we also considered using SMT solvers for the process of extracting a consistent DFA from the learned 3DFA. This approach constructs a set of constraints to check whether there exists a DFA of size $k$ that is consistent with the 3DFA. This is achieved by a function that maps each state in the 3DFA to a state in the DFA. The goal of the solver is to find such a function that results in a consistent DFA. A minimal consistent DFA is found using a binary search over the potential DFA size interval, ranging from one to the size of the 3DFA. Our implementation of this method presented significant computational challenges without offering improvements over RPNI. Therefore, we decided not to include it in our report and focused solely on the RPNI algorithm. However, we find this approach promising, as the solver's flexibility allows for the incorporation of additional constraints that can refine the resulting DFA. For instance, one could minimize the number of non-self transitions to obtain a more interpretable DFA. While this method was not practical for our current study due to scalability issues, further research is needed to explore its potential improvements and applicability in larger-scale settings. We leave this investigation for future work.

Our evaluation focused on a set of real-world Java control software benchmarks. While these benchmarks demonstrate the feasibility of our approach and its applicability to realistic systems with varying complexity, they may not capture the full spectrum of complexities found in broader scenarios. Nonetheless, we believe they are reasonably representative for the scope of this paper and plan to study additional examples in future work.

\section{Related Work}
\label{sec:related-work}

Error explanations aim to help users understand the core cause of a failure and how to potentially correct it. Several approaches for generating meaningful explanations have been suggested in recent years. An approach with the most resemblance to ours was proposed by Chockler et al.~\cite{chockler_learning_2020}, who introduced a method for explaining software errors by learning a DFA that captures the faulty behavior using the $L^*$ algorithm. In the context of our work, they suggest learning the language $B$ when $T=\Sigma^*$. Our work, in some ways, extends this approach by proposing models that abstract irrelevant information typically included in $B$, providing general models that focus solely on the faulty behavior.

Some error-explanation techniques rely on counterexample-based analysis to provide insights about the errors. Notable examples include methods featured in the SLAM model checker~\cite{ball_symptom_2003} and Java Pathfinder~\cite{groce_what_2003}. These techniques typically identify multiple versions of an error, along with similar execution traces that do not lead to a failure. By analyzing these executions, they extract and highlight the key elements of the error. Sharygina and Peled~\cite{sharygina_combined_2001} suggested an exploration technique of the ``neighborhood'' of a counterexample to better understand the error. Groce et al.~\cite{groce_error_2006} proposed an improved automated approach for understanding and isolating errors based on distance metrics for program executions. In another work, Aboussoror et al.~\cite{aboussoror_seeing_2012} propose Metaviz, a model-driven framework for customizing simulation trace visualizations in real-time embedded systems.

Mayer and Stumptner survey model-based software debugging methods~\cite{DBLP:conf/kbse/MayerS08}, which localize faults by constructing a program model from test cases and identifying discrepancies with expected outcomes. Unlike traditional techniques that rely on test case datasets, our approach takes a model of the test space, $T$, as input (although it may utilize test cases to enhance results or performance). This method offers increased flexibility, enabling testers to both over- and under-approximate various aspects of the test space. Rather than identifying faulty code, it provides abstract explanations of potentially problematic behavioral patterns.

The framework of active automata learning from an incomplete teacher has been extensively explored in recent years~\cite{grinchtein_inferring_2006,chen_learning_2009,moeller_automata_2023}. Grinchtein et al.~\cite{grinchtein_inferring_2006} presented an adaptation of $L^*$ to handle missing information with a SAT solver for inferring network invariants. Chen et al.~\cite{chen_learning_2009} considered an improved algorithm for learning in this framework to identify a minimal separating DFA between two regular languages for compositional verification. Their approach involves initially computing a 3DFA and subsequently transforming it into a minimal consistent DFA using the algorithm outlined in~\cite{paull_minimizing_1959}. However, this algorithm is exponential in complexity, whereas the RPNI algorithm, which we employ here for this conversion, runs in polynomial time.

In a subsequent work, Leucker and Neider~\cite{leucker_learning_2012} presented this framework and surveyed some learners that operate within it. These include a naive approach that enumerates all DFAs of increasing size, the approach of~\cite{grinchtein_inferring_2006}, as well as the approach of~\cite{chen_learning_2009} that we described earlier. They do not explore any implementation aspects or benchmarking. To the best of our knowledge, the work presented here is the first to discuss this framework for error explanations.

\section{Conclusion and Future Work}
\label{sec:conclusion}

We presented a framework that uses automata to generate concise bug descriptions. Specifically, we introduced \emph{Failure Explanations (FE)}, \emph{Eventual Failure Explanations (EFE)}, and \emph{Early Detection (ED)}, concepts that highlight essential failure patterns by filtering out irrelevant behavior. Evaluation on real-world benchmarks shows that these techniques yield significantly more compact and focused descriptions than the original traces.
Motivated by our prior work on testing with behavioral programming, where bugs were often hard to interpret, we plan to investigate how the proposed explanations can enhance practitioner understanding in that context.


\section*{Acknowledgements}

This research was supported by the Lynne and William Frankel Center for Computer Science at Ben-Gurion University and by funding from the Israel Science Foundation (ISF), grant number 2714/19.

\bibliographystyle{IEEEtran}
\bibliography{more,references}

\begin{thebibliography}{10}
\providecommand{\url}[1]{#1}
\csname url@samestyle\endcsname
\providecommand{\newblock}{\relax}
\providecommand{\bibinfo}[2]{#2}
\providecommand{\BIBentrySTDinterwordspacing}{\spaceskip=0pt\relax}
\providecommand{\BIBentryALTinterwordstretchfactor}{4}
\providecommand{\BIBentryALTinterwordspacing}{\spaceskip=\fontdimen2\font plus
\BIBentryALTinterwordstretchfactor\fontdimen3\font minus
  \fontdimen4\font\relax}
\providecommand{\BIBforeignlanguage}[2]{{%
\expandafter\ifx\csname l@#1\endcsname\relax
\typeout{** WARNING: IEEEtran.bst: No hyphenation pattern has been}%
\typeout{** loaded for the language `#1'. Using the pattern for}%
\typeout{** the default language instead.}%
\else
\language=\csname l@#1\endcsname
\fi
#2}}
\providecommand{\BIBdecl}{\relax}
\BIBdecl

\bibitem{bar-sinai_provengo_2023}
M.~Bar-Sinai, A.~Elyasaf, G.~Weiss, and Y.~Weiss, ``Provengo: {A} {Tool}
  {Suite} for {Scenario} {Driven} {Model}-{Based} {Testing},'' in \emph{2023
  38th {IEEE}/{ACM} {International} {Conference} on {Automated} {Software}
  {Engineering} ({ASE})}, 2023, pp. 2062--2065.

\bibitem{angluin_learning_1987}
D.~Angluin, ``Learning regular sets from queries and counterexamples,''
  \emph{Information and computation}, vol.~75, no.~2, pp. 87--106, 1987,
  publisher: Elsevier.

\bibitem{oncina_identifying_1992}
J.~Oncina and P.~Garcia, ``Identifying regular languages in polynomial time,''
  in \emph{Advances in structural and syntactic pattern recognition}.\hskip 1em
  plus 0.5em minus 0.4em\relax World Scientific, 1992, pp. 99--108.

\bibitem{de_la_higuera_grammatical_2010}
C.~De~la Higuera, \emph{Grammatical inference: learning automata and
  grammars}.\hskip 1em plus 0.5em minus 0.4em\relax Cambridge University Press,
  2010.

\bibitem{8117878}
P.~Ramya, V.~Sindhura, and P.~V. Sagar, ``Testing using selenium web driver,''
  in \emph{2017 Second International Conference on Electrical, Computer and
  Communication Technologies (ICECCT)}, 2017, pp. 1--7.

\bibitem{10714759}
K.~Vadia, A.~Thukrul, P.~G. Mazumdar, N.~Panda, and A.~Sirsat, ``Bug testing
  automation with playwright and a backend api,'' in \emph{2024 8th
  International Conference on I-SMAC (IoT in Social, Mobile, Analytics and
  Cloud) (I-SMAC)}, 2024, pp. 1867--1870.

\bibitem{utting2012practical}
M.~Utting and B.~Legeard, \emph{Practical model-based testing: a tools
  approach}.\hskip 1em plus 0.5em minus 0.4em\relax Morgan Kaufmann, 2007.

\bibitem{lee1996principles}
D.~Lee and M.~Yannakakis, ``Principles and methods of testing finite state
  machines-a survey,'' \emph{Proceedings of the IEEE}, vol.~84, no.~8, pp.
  1090--1123, 1996.

\bibitem{tretmans2008model}
\BIBentryALTinterwordspacing
J.~Tretmans, \emph{Model Based Testing with Labelled Transition Systems}.\hskip
  1em plus 0.5em minus 0.4em\relax Berlin, Heidelberg: Springer Berlin
  Heidelberg, 2008, pp. 1--38. [Online]. Available:
  \url{https://doi.org/10.1007/978-3-540-78917-8_1}
\BIBentrySTDinterwordspacing

\bibitem{supplementarymaterial}
T.~Yaacov, G.~Weiss, G.~Amram, and A.~Hayoun, ``Automata models for effective
  bug pattern description --- supplementary material,''
  \url{https://doi.org/10.5281/zenodo.15798674}, 2025.

\bibitem{chen_learning_2009}
Y.-F. Chen, A.~Farzan, E.~M. Clarke, Y.-K. Tsay, and B.-Y. Wang, ``Learning
  minimal separating {DFA}’s for compositional verification,'' in
  \emph{International {Conference} on {Tools} and {Algorithms} for the
  {Construction} and {Analysis} of {Systems}}.\hskip 1em plus 0.5em minus
  0.4em\relax Springer, 2009, pp. 31--45.

\bibitem{ammann2017introduction}
P.~Ammann and J.~Offutt, \emph{Introduction to software testing}.\hskip 1em
  plus 0.5em minus 0.4em\relax Cambridge University Press, 2017.

\bibitem{leucker_learning_2012}
M.~Leucker and D.~Neider, ``Learning minimal deterministic automata from
  inexperienced teachers,'' in \emph{International {Symposium} {On}
  {Leveraging} {Applications} of {Formal} {Methods}, {Verification} and
  {Validation}}.\hskip 1em plus 0.5em minus 0.4em\relax Springer, 2012, pp.
  524--538.

\bibitem{1702519}
T.~Chow, ``Testing software design modeled by finite-state machines,''
  \emph{IEEE Transactions on Software Engineering}, vol. SE-4, no.~3, pp.
  178--187, 1978.

\bibitem{533956}
D.~Lee and M.~Yannakakis, ``Principles and methods of testing finite state
  machines-a survey,'' \emph{Proceedings of the IEEE}, vol.~84, no.~8, pp.
  1090--1123, 1996.

\bibitem{paull_minimizing_1959}
M.~C. Paull and S.~H. Unger, ``Minimizing the number of states in incompletely
  specified sequential switching functions,'' \emph{IRE Transactions on
  Electronic Computers}, no.~3, pp. 356--367, 1959, publisher: IEEE.

\bibitem{muskardin_aalpy_2022}
\BIBentryALTinterwordspacing
E.~Muškardin, B.~K. Aichernig, I.~Pill, A.~Pferscher, and M.~Tappler,
  ``{AALpy}: an active automata learning library,'' \emph{Innovations in
  Systems and Software Engineering}, vol.~18, no.~3, pp. 417--426, Sep. 2022.
  [Online]. Available: \url{https://doi.org/10.1007/s11334-022-00449-3}
\BIBentrySTDinterwordspacing

\bibitem{jasper_rers_2019}
M.~Jasper, M.~Mues, A.~Murtovi, M.~Schlüter, F.~Howar, B.~Steffen,
  M.~Schordan, D.~Hendriks, R.~Schiffelers, H.~Kuppens, and F.~W. Vaandrager,
  ``{RERS} 2019: {Combining} {Synthesis} with {Real}-{World} {Models},'' in
  \emph{Tools and {Algorithms} for the {Construction} and {Analysis} of
  {Systems}}, D.~Beyer, M.~Huisman, F.~Kordon, and B.~Steffen, Eds.\hskip 1em
  plus 0.5em minus 0.4em\relax Cham: Springer International Publishing, 2019,
  pp. 101--115.

\bibitem{havelund_model_2000}
\BIBentryALTinterwordspacing
K.~Havelund and T.~Pressburger, ``Model checking {JAVA} programs using {JAVA}
  {PathFinder},'' \emph{International Journal on Software Tools for Technology
  Transfer}, vol.~2, no.~4, pp. 366--381, Mar. 2000. [Online]. Available:
  \url{https://doi.org/10.1007/s100090050043}
\BIBentrySTDinterwordspacing

\bibitem{chockler_learning_2020}
H.~Chockler, P.~Kesseli, D.~Kroening, and O.~Strichman, ``Learning the language
  of software errors,'' \emph{Journal of Artificial Intelligence Research},
  vol.~67, pp. 881--903, 2020.

\bibitem{ball_symptom_2003}
\BIBentryALTinterwordspacing
T.~Ball, M.~Naik, and S.~K. Rajamani, ``From symptom to cause: localizing
  errors in counterexample traces,'' in \emph{Proceedings of the 30th {ACM}
  {SIGPLAN}-{SIGACT} {Symposium} on {Principles} of {Programming} {Languages}},
  ser. {POPL} '03.\hskip 1em plus 0.5em minus 0.4em\relax New York, NY, USA:
  Association for Computing Machinery, 2003, pp. 97--105, event-place: New
  Orleans, Louisiana, USA. [Online]. Available:
  \url{https://doi.org/10.1145/604131.604140}
\BIBentrySTDinterwordspacing

\bibitem{groce_what_2003}
A.~Groce and W.~Visser, ``What {Went} {Wrong}: {Explaining}
  {Counterexamples},'' in \emph{Model {Checking} {Software}}, T.~Ball and S.~K.
  Rajamani, Eds.\hskip 1em plus 0.5em minus 0.4em\relax Berlin, Heidelberg:
  Springer Berlin Heidelberg, 2003, pp. 121--136.

\bibitem{sharygina_combined_2001}
N.~Sharygina and D.~Peled, ``A {Combined} {Testing} and {Verification}
  {Approach} for {Software} {Reliability},'' in \emph{{FME} 2001: {Formal}
  {Methods} for {Increasing} {Software} {Productivity}}, J.~N. Oliveira and
  P.~Zave, Eds.\hskip 1em plus 0.5em minus 0.4em\relax Berlin, Heidelberg:
  Springer Berlin Heidelberg, 2001, pp. 611--628.

\bibitem{groce_error_2006}
\BIBentryALTinterwordspacing
A.~Groce, S.~Chaki, D.~Kroening, and O.~Strichman,
  ``\BIBforeignlanguage{en}{Error explanation with distance metrics},''
  \emph{\BIBforeignlanguage{en}{International Journal on Software Tools for
  Technology Transfer}}, vol.~8, no.~3, pp. 229--247, Jun. 2006. [Online].
  Available: \url{https://doi.org/10.1007/s10009-005-0202-0}
\BIBentrySTDinterwordspacing

\bibitem{aboussoror_seeing_2012}
E.~A. Aboussoror, I.~Ober, and I.~Ober, ``Seeing {Errors}: {Model} {Driven}
  {Simulation} {Trace} {Visualization},'' in \emph{Model {Driven} {Engineering}
  {Languages} and {Systems}}, R.~B. France, J.~Kazmeier, R.~Breu, and
  C.~Atkinson, Eds.\hskip 1em plus 0.5em minus 0.4em\relax Berlin, Heidelberg:
  Springer Berlin Heidelberg, 2012, pp. 480--496.

\bibitem{DBLP:conf/kbse/MayerS08}
\BIBentryALTinterwordspacing
W.~Mayer and M.~Stumptner, ``Evaluating models for model-based debugging,'' in
  \emph{23rd {IEEE/ACM} International Conference on Automated Software
  Engineering {(ASE} 2008), 15-19 September 2008, L'Aquila, Italy}.\hskip 1em
  plus 0.5em minus 0.4em\relax {IEEE} Computer Society, 2008, pp. 128--137.
  [Online]. Available: \url{https://doi.org/10.1109/ASE.2008.23}
\BIBentrySTDinterwordspacing

\bibitem{grinchtein_inferring_2006}
O.~Grinchtein, M.~Leucker, and N.~Piterman, ``Inferring network invariants
  automatically,'' in \emph{International {Joint} {Conference} on {Automated}
  {Reasoning}}.\hskip 1em plus 0.5em minus 0.4em\relax Springer, 2006, pp.
  483--497.

\bibitem{moeller_automata_2023}
\BIBentryALTinterwordspacing
M.~Moeller, T.~Wiener, A.~Solko-Breslin, C.~Koch, N.~Foster, and A.~Silva,
  ``Automata {Learning} with an {Incomplete} {Teacher},'' in \emph{37th
  {European} {Conference} on {Object}-{Oriented} {Programming} ({ECOOP} 2023)},
  ser. Leibniz {International} {Proceedings} in {Informatics} ({LIPIcs}),
  K.~Ali and G.~Salvaneschi, Eds., vol. 263.\hskip 1em plus 0.5em minus
  0.4em\relax Dagstuhl, Germany: Schloss Dagstuhl, 2023, pp. 21:1--21:30, iSSN:
  1868-8969. [Online]. Available:
  \url{https://drops.dagstuhl.de/entities/document/10.4230/LIPIcs.ECOOP.2023.21}
\BIBentrySTDinterwordspacing

\end{thebibliography}
\end{document}